\documentclass[a4paper,UKenglish, thm-restate,numberwithinsect]{lipics-v2021}
\title{First-Order Logic with Connectivity Operators} %TODO Please add

\usepackage[all,defaultlines=3]{nowidow}

\titlerunning{First-Order Logic with Connectivity Operators} %TODO optional, please use if title is longer than one line

\author{Nicole Schirrmacher}{University of Bremen, Germany}{schirrmacher@uni-bremen.de}{https://orcid.org/0000-0002-1740-7478}{}%TODO mandatory, please use full name; only 1 author per \author macro; first two parameters are mandatory, other parameters can be empty. Please provide at least the name of the affiliation and the country. The full address is optional

\author{Sebastian Siebertz}{University of Bremen, Germany}{siebertz@uni-bremen.de}{https://orcid.org/0000-0002-6347-1198}{}

\author{Alexandre Vigny}{University of Bremen, Germany}{vigny@uni-bremen.de}{https://orcid.org/0000-0002-4298-8876}{}

\authorrunning{N.\ Schirrmacher, S.\ Siebertz and A.\ Vigny} %TODO mandatory. First: Use abbreviated first/middle names. Second (only in severe cases): Use first author plus 'et al.'

\Copyright{Nicole Schirrmacher, Sebastian Siebertz and Alexandre Vigny} %TODO mandatory, please use full first names. LIPIcs license is "CC-BY";  http://creativecommons.org/licenses/by/3.0/


 %TODO mandatory; please add comma-separated list of keywords

\category{} %optional, e.g. invited paper

\relatedversion{} %optional, e.g. full version hosted on arXiv, HAL, or other respository/website
\funding{This paper is a part of the ANR-DFG project \emph{Unifying Theories for Multivariate Algorithms} (\rm{UTMA}), which has
received funding from the German
Research Foundation (DFG) with grant agreement
No 446200270.}%optional, to capture a funding statement, which applies to all authors. Please enter author specific funding statements as fifth argument of the \author macro.

\acknowledgements{We thank Mikołaj Boja\'nczyk for fruitful discussions. He independently
studied $\FOconn$ and suggested the name \emph{separator logic} for
this logic. We also thank Michał Pilipczuk and Szymon Toru\'nczyk for
fruitful discussions.}%optional

\nolinenumbers %uncomment to disable line numbering

\EventEditors{Florin Manea and Alex Simpson}
\EventNoEds{2}
\EventLongTitle{30th EACSL Annual Conference on Computer Science Logic (CSL 2022)}
\EventShortTitle{CSL 2022}
\EventAcronym{CSL}
\EventYear{2022}
\EventDate{February 14--19, 2022}
\EventLocation{G\"{o}ttingen, Germany (Virtual Conference)}
\EventLogo{}
\SeriesVolume{216}
\ArticleNo{20}
\usepackage{mathrsfs}
\usepackage{tikz}
\usetikzlibrary{calc}
\usepackage{url}
\usepackage{hyperref}
\hypersetup{
    colorlinks = true,
    linkcolor = blue,
    citecolor=blue
    }
\usepackage[noabbrev,capitalise,nameinlink]{cleveref}
\newcommand\claimref[2]{\hyperref[{#1}]{#2}}

\usepackage{mathtools}

\newcommand{\FO}{\mathrm{FO}}
\newcommand{\Conn}{\mathrm{conn}}
\newcommand{\FOconn}{\mathrm{FO}+\mathrm{conn}}
\newcommand{\TC}{\mathrm{TC}}
\newcommand{\MSO}{\mathrm{MSO}}
\newcommand{\Cc}{\mathscr{C}}
\newcommand{\Ff}{\mathcal{F}}
\newcommand{\DP}{\mathrm{DP}}
\newcommand{\FODP}{\mathrm{FO+DP}}

\newcommand{\del}{\mathrm{del}}
\newcommand{\disjp}{\mathrm{disjoint}\text{-}\mathrm{paths}}
\newcommand{\conn}{\mathrm{conn}}

\renewcommand{\phi}{\varphi}
\newcommand{\minor}{\preccurlyeq}
\newcommand{\strA}{\mathfrak{A}}
\newcommand{\strB}{\mathfrak{B}}

\begin{document}

\maketitle

% !TEX root = main.tex

\begin{abstract}
  First-order logic ($\FO$) can express many algorithmic problems on graphs,
  such as the independent set and dominating set problem parameterized
  by solution size. On the other hand, $\FO$ cannot
  express the very simple algorithmic question whether two vertices
  are connected. We enrich $\FO$ with connectivity
  predicates that are tailored to express algorithmic graph properties that
  are commonly studied in parameterized algorithmics. By adding
  the atomic predicates
  $\conn_k(x,y,z_1,\ldots, z_k)$ that hold true in a graph if there exists a
  path between (the valuations of) $x$ and $y$ after (the valuations of)
  $z_1,\ldots, z_k$ have been deleted, we obtain \emph{separator logic}
  $\FOconn$. We show that separator logic
  can express many interesting problems such as the feedback vertex
  set problem and elimination distance problems to first-order
  definable classes. Denote by $\FOconn_k$ the fragment of
  separator logic that is restricted to connectivity predicates with
  at most $k+2$ variables (that is, at most $k$ deletions).
  We show that $\FOconn_{k+1}$ is strictly
  more expressive than $\FOconn_k$ for all $k\geq 0$.
We then study the limitations of separator logic
  and prove that it cannot express planarity, and, in particular, not the
  disjoint paths problem. We obtain the stronger \emph{disjoint-paths
  logic} $\FODP$ by adding the atomic predicates
  $\disjp_k[(x_1,y_1),\ldots, (x_k,y_k)]$ that evaluate to true if there are
  internally vertex-disjoint paths between (the valuations of) $x_i$ and
  $y_i$ for all~$1\leq i\leq k$. Disjoint-paths logic can express the disjoint
  paths problem, the problem
  of (topological) minor containment, the problem of hitting (topological)
  minors, and many more. Again we show that the fragments $\FODP_k$
  that use predicates for at most $k$ disjoint paths form a strict
  hierarchy of expressiveness.
  Finally, we compare the expressive power of the
  new logics with that of transitive-closure logics and monadic second-order logic.
\end{abstract}

\pagebreak
% !TEX root = main.tex

\section{Introduction}
\label{sec:introduction}

Logic provides a very elegant way of formally describing computational
problems. Fagin's celebrated result in 1974~\cite{fagin1974generalized} established that
existential second-order logic captures the complexity class \textsc{NP}. Fagin
thereby provided a machine-independent characterization of a complexity
class and initiated the field of descriptive complexity theory. Many other
complexity classes were later characterized by logics in this theory. Today
it remains one of the major open problems whether there exists a logic
capturing \textsc{PTime}.

In 1990 Courcelle proved that every graph property definable in
monadic second-order logic (MSO) can be decided in linear time on graphs of
bounded treewidth~\cite{courcelle1990monadic}. This theorem
has a much more algorithmic (rather than a complexity-theoretic) flavor,
in the sense that, from a logical description of a problem, it derives an
algorithmic approach on how to solve it on certain graph classes.
Grohe in his seminal survey coined the term \emph{algorithmic
meta-theorem} for such  theorems that provide general conditions
on a problem and on the input instances that, when satisfied, imply
the existence of an efficient algorithm for the problem~\cite{grohe2008logic}.
Courcelle's theorem for $\MSO$ was extended
to graph classes with bounded cliquewidth~\cite{courcelle2000linear}
and it is known that these are essentially the most general graph classes
on which efficient MSO model-checking~\cite{ganian2014lower,kreutzer2010lower}
is possible.
MSO is a powerful logic that can express many important algorithmic
properties on graphs. With quantification over edges, we can for example
express the existence of a Hamiltonian path, the existence of a fixed minor
or topological minor, the disjoint paths problem, and many deletion
problems. For a property~$\Pi$, the task in the $\Pi$-deletion problem
is to find in a given graph $G$ a minimum-size subset $S$ of $V(G)$ such that
the graph $G-S$ obtained from $G$ by removing $S$ has the property~$\Pi$. Important examples of $\Pi$-deletion problems
are the feedback vertex set problem, the
odd cycle transversal problem, or the problem of hitting all minors or
topological minors from a given list $\Ff$. Also, many elimination distance
problems recently studied~\cite{bulian2016graph} in parameterized algorithmics can be expressed
in MSO. However, as we have seen, this expressiveness comes at the
price of algorithmic intractability already on very restricted graph classes.
This cannot be a surprise as e.g.\ the Hamiltonian path problem is NP-complete
already on planar graphs of maximum degree $3$~\cite{buro2000simple}.
%\alex{I don't find exactly that in~\cite{garey1974some}. I find :
%\begin{enumerate}
%  \item 3 colorability VS planar + degree at most 4.
%  \item Undirected Hamiltonian circuit VS degree at most 3
%  \item directed Hamiltonian path VS planar + in-degree $\le 3$ , out-degree $\le 4$
%\end{enumerate}
%See thm 2.1, 2.2, 2.3 here \url{https://sci-hub.st/10.1145/800119.803884}}

First-order logic (FO) is much weaker than MSO and consequently, the
model-checking problem can be solved efficiently on much more general
graph classes. FO model-checking is fixed-parameter tractable on a
subgraph-closed class $\Cc$ if and only
if $\Cc$ is nowhere dense~\cite{grohe2017deciding} and a recent
breakthrough result showed that it is
fixed-parameter tractable on a class $\Cc$ of ordered graphs if and only
if $\Cc$ has bounded twin-width~\cite{bonnet2021twin}. FO is weaker
than MSO but it can still express many important problems
such as the independent set problem and dominating set problem parameterized
by solution size, the Steiner tree problem parameterized by the number of Steiner vertices, and many more problems. On the other hand, first-order logic
cannot even express the algorithmically extremely simple problem of
whether a graph is connected. Also, the other algorithmic problems
mentioned before are not expressible in FO, even though some of
them are fixed-parameter tractable on general graphs. For example,
we can efficiently test for a fixed minor or topological minor and
solve the disjoint paths problem~\cite{robertson1995graph}. Many
$\Pi$-deletion problems are fixed-parameter tractable, see
e.g.~\cite{book15parameterized-algorithms,fomin2020hitting,reed2004finding},
as well as many elimination distance problems~\cite{agrawal2021fpt,FominGST20}.

The fact that first-order logic can only express local properties is classically
addressed by adding transitive-closure or fixed-point operators,
see e.g.~\cite{ebbinghaus2005finite,gradel2007finite,libkin2013elements}. Unfortunately,
this again comes at the price of intractable model-checking for very restricted
graph classes. For example, even the model-checking problem for
the very restricted monadic transitive-closure logic $\mathrm{TC}^1$
studied by Grohe~\cite{grohe2008logic}, is $\mathrm{AW}[\star]$-hard
on planar graphs of maximum degree at most~3~\cite[Theorem~7.3]{grohe2008logic}.
Also, these logics fall short of being able
to express all of the above mentioned algorithmic graph problems studied in recent
parameterized algorithmics.

This motivates our present work in which we enrich first-order logic
with basic connectivity predicates. The extensions are tailored
to express algorithmic graph properties that are studied in recent
parameterized algorithmics. We can add the atomic predicate
$\conn_0(x,y)$ that evaluates to true on a graph $G$
if (the valuations of) $x$ and $y$ are connected in $G$. This
predicate easily generalizes to directed graphs but for simplicity,
we work with undirected graphs only. Of course, with this predicate we
can express connectivity of graphs, however, it falls short of
expressing other interesting properties, e.g.\ %it is easy to see that
it cannot express that a graph is acyclic. We hence introduce
more general predicates $\conn_k(x,y,z_1,\ldots, z_k)$,
parameterized by a number $k$, that evaluate to true on a graph $G$
if (the valuations of) $x$ and~$y$ are connected in $G$ once
(the valuations of) $z_1,\ldots, z_k$ have been deleted.
%a second operator $\del(x)[\phi]$, where $\phi$ is a formula
%that does not have $x$ as a free variable, that evaluates $\phi$ with
%(the valuation of) $x$ deleted.
The interplay of these predicates with the usual nesting of
first-order quantification
makes the new logic $\FOconn$ already quite powerful.
For example, we can express simple properties such as
%$2$-connectivity by $\forall x\big(\del(x)[\forall y\forall z( \disjp[(y,z)]\big)$.
$2$-connectivity by $\forall z \forall x\forall y\bigl(x\neq z \wedge y\neq z \rightarrow \conn_1(x,y,z)\bigr)$.
We can also express many deletion problems, such as the feedback
vertex set problem, and the elimination distance to bounded degree,
and more generally, elimination distance to any fixed first-order property.

We also point to the work of Mikołaj Boja\'nczyk~\cite{bojanczyk2021separator}, who independently introduced
\mbox{$\FOconn$} and proposed the name \emph{separator logic}.
He studied a variant of star-free expressions for graphs and showed that
these expressions exactly correspond to separator logic. We follow
his suggestion and thank Mikołaj for the discussion on
separator logic.

In \cref{sec:foconn} we study the expressive power of separator logic.
We give
examples on properties expressible with separator logic as well as
proofs that certain properties, such as planarity and in particular
the disjoint paths problem, are not expressible in separator logic.
%A natural parameter for separator formulas is the nesting depth
%of deletion operators, which we call \emph{deletion depth}.
We show that $(k+2)$-connectivity of a graph cannot be expressed with
only $\conn_k$ predicates and conclude that the restricted use of
these predicates
induces a natural hierarchy of expressiveness.
% Finally, we compare
% separator logic to transitive closure and fixed-point logics.

Using the notion of {\em block decompositions} together with known model-checking
results, one can show that model-checking for formulas using only $\conn_1$
predicates is fixed-parameter tractable on nowhere dense classes
of graphs. Hence, we can evaluate very simple connectivity queries in
formulas without an increase in the complexity of the model-checking
problem on subgraph-closed graph classes. On the other hand, when
we allow $\conn_2$ predicates, there are some simple graph classes that do not
exclude a topological minor, and on which model-checking becomes
$\mathrm{AW}[\star]$-hard.
In this paper,
we do not go into the details of model-checking, but in a companion
paper~\cite{pilipczuk2021algorithms}, we prove that in fact model-checking for
\mbox{$\FOconn$} %on subgraph closed classes
is fixed-parameter tractable on graph classes that exclude a topological
minor.

%We also point to work of Boja\'nczyk, who independently introduced
%\mbox{$\FOconn+\del$} and proposed the name separator logic.
%He studies a variant of star-free expressions for graphs and shows that
%these expressions exactly correspond to separator logic. He uses a
%slightly different syntax, which however leads to a logic with the
%same expressive power. Instead of separating connectivity and deletion
%of elements as we do, he uses an atomic predicate $\conn_k(x,y,z_1,\ldots, z_k)$
%that evaluates to true if (the valuations of) $x$ and $y$ are connected after
%the deletion of $z_1,\ldots, z_k$. It is easy to see that every formula
%of deletion depth~$k$ can be written as an FO formula with
%the $\conn_k$ predicate and vice versa.

The fact that planarity and the disjoint paths problem
cannot be expressed in separator logic motivates us to
define an even stronger logic that can express these
properties. The atomic predicate $\disjp_k[(x_1,y_1),\ldots,
(x_k, y_k)]$ evaluates to true if and only if there are internally
vertex-disjoint paths between (the valuations of)
$x_i$ and $y_i$ for all~$1\leq i\leq k$.
Connectivity of $x$ and $y$
can be tested by $\disjp_1[(x,y)]$.
% Also, the $\del$ operator
% can be simulated with the help of first-order quantifiers and
% $\disjp$, hence, the logic $\FO+\disjp$ strictly extends
% $\FOconn$.
More generally, the so obtained \emph{disjoint-paths logic}
$\FODP$ strictly extends separator logic.
With this more powerful logic, we can test if a graph contains a
fixed minor or topological minor, and in particular, test for
planarity. In combination with first-order quantification, we can
also express many $\Pi$-deletion problems such as the problem
of hitting all minors or topological minors from a given list $\Ff$.
On the other hand, we cannot express the odd cycle transversal
problem, as we cannot even express bipartiteness of a graph.
We study the expressive power of $\FODP$ in \cref{sec:fodp}.
Among other results, we prove that again an increase in the number
of disjoint paths in the predicates leads to an increase in
expressive power.

Note that while it would be desirable to be able to express
bipartiteness, which is equivalent to $2$-colorability,
it is not desirable to express general colorability problems, as
we aim for logics that are tractable on planar graphs and beyond,
while the $3$-colorability problem is NP-complete on planar graphs.
This example shows again that it is a delicate balance between
expressiveness and tractability and it will be a challenging
and highly interesting problem in future work to find the
right set of predicates to express even more algorithmic
graph properties while at the same time having tractable
model-checking. Until now the complexity of the model-checking
problem for $\FODP$ has remained elusive and will be a
very interesting problem in future work.

We conclude the paper in \cref{sec:other-logics} with a comparison between the
newly introduced logics and more established ones, like MSO and transitive-closure logics.

% !TEX root = main.tex

\section{Preliminaries}
\label{sec:preliminaries}

\textbf{Graphs.} In this paper we deal with finite and simple undirected
graphs.
%Most results will hold also for directed graphs or more generally,
%for binary relational structures, however, for simplicity we focus on
%graphs.
Let $G$ be a graph. We write $V(G)$
for the vertex set of~$G$ and~$E(G)$ for
its edge set. For a
set $X\subseteq V(G)$ we write $G[X]$ for the subgraph of $G$ induced by
$X$ and $G-X$ for the subgraph induced by $V(G)\setminus X$. For a
singleton set~$\{v\}$ we write~$G-v$ instead of $G-\{v\}$. A {\em path} $P$ in
$G$ is a subgraph on distinct vertices $v_1,\ldots, v_t$ with $\{v_i,v_{i+1}\}\in E(P)$ for all $1\leq i<t$ and a path~$P$ is said to {\em connect} its endpoints~$v_1$ and~$v_t$.
Two paths are {\em internally vertex-disjoint} if and only if every vertex that appears in both paths is an end point of both paths.
The graph~$G$ is {\em connected} if every two of its vertices are connected by a
path. It is {\em $k$-connected} if $G$ has more than $k$ vertices and $G-X$ is connected for every subset $X\subseteq V(G)$ of size strictly smaller than $k$. A {\em cycle}~$C$ in $G$ is a subgraph on
distinct vertices $v_1,\ldots, v_t$, $t\geq 3$, with $\{v_t,v_1\}\in E(C)$
and  $\{v_i,v_{i+1}\}\in E(C)$ for all $1\leq i<t$. An acyclic graph is a
{\em forest} and a connected acyclic graph is a {\em tree}.
%The graph $G$ is {\em Eulerian}
%if it is connected and all its vertices have even degree and it is
%{\em Hamiltonian} if it contains a cycle on all vertices as a subgraph. \sebi{Do we have the results on Eulerian and Hamiltonian graphs?}

A graph $H$ is a {\em minor} of $G$, denoted $H\minor G$,  if for all $v\in V(H)$
there are pairwise vertex-disjoint connected subgraphs $G_v$ of $G$
such that whenever $\{u,v\}\in E(H)$, then there are $x\in V(G_u)$
and $y\in V(G_v)$ with $\{x,y\}\in E(G)$. The graph $H$ is a {\em topological
minor} of~$G$, denoted $H\minor^{top}G$, if for all $v\in V(H)$ there
is a distinct vertex $x_v$ in $G$ and for all $\{u,v\}\in E(H)$ there
are internally vertex-disjoint paths $P_{uv}$ in $G$ with endpoints
$x_u$ and $x_v$. A graph is {\em planar} if and only if it does not contain $K_5$,
the complete graph on $5$ vertices, and $K_{3,3}$, the complete
bipartite graph with two partitions of size $3$, as a minor.

\medskip\noindent\textbf{Logic.} In this work we deal with structures
over purely relational {\em signatures}. A (purely relational) signature is a collection
of relation symbols, each with an associated arity. Let $\sigma$ be a
signature. A {\em $\sigma$-structure} $\mathfrak{A}$ consists of a non-empty
set $A$, the universe of~$\mathfrak{A}$, together with an interpretation of
each $k$-ary relation symbol $R\in\sigma$ as a $k$-ary relation
$R^\mathfrak{A}\subseteq A^k$. For a subset $X\subseteq A$ we write
$\strA[X]$ for the substructure induced by $X$. A {\em partial isomorphism}
between $\sigma$-structures $\strA$ and $\strB$ is an isomorphism
between $\strA[X]$ and $\strB[Y]$ for some subset $X\subseteq A$ of the
universe $A$ of $\strA$ and some subset $Y\subseteq B$ of the universe $B$ of
$\strB$.

%  $\overrightarrow{a} = (a_1,\ldots,a_n)$ and $\overrightarrow{b} = (b_1,\ldots,b_n)$ two tuples in $\mathfrak{A}$ and $\mathfrak{B}$ respectively. Then, $\left(\overrightarrow{a}, \overrightarrow{b}\right)$ defines a \textit{partial isomorphism} between $\mathfrak{A}$ and $\mathfrak{B}$ if the following conditions hold:
%\begin{enumerate}
%\item For every $i,j\leq n$, \[a_i=a_j\hspace{2ex}\text{ iff }\hspace{2ex}b_i=b_j\]
%\item For every constant symbol $c\in\sigma$, and every $i\leq n$,
%\[a_i=c^\mathfrak{A}\hspace{2ex}\text{ iff }\hspace{2ex}b_i=c^\mathfrak{B}\]
%\item For every $m$-ary relation symbol $R\in\sigma$ and every sequence $(i_1,\ldots,i_m)$ of (not necessarily distinct) numbers from $[1, n]$,
%\[(a_{i1},\ldots,a_{im})\in R^\mathfrak{A} \hspace{2ex}\text{ iff }\hspace{2ex} (b_{i1},\ldots,b_{im})\in R^\mathfrak{B}\]
%\end{enumerate}
%If $\sigma$ does not contain any constant, a partial isomorphism is the mapping $a_i\mapsto b_i,i\leq n$ that is an isomorphism between the substructures of $\mathfrak{A}$ and $\mathfrak{B}$ generated by $\{a_1,\ldots,a_n\}$ and $\{b_1,\ldots,b_n\}$.~\cite{libkin2013elements}

%\subsection{First-Order Logic}

We assume an infinite supply $\textsc{Var}$ of variables.
First-order formulas are built from the atomic formulas $x=y$, where
$x$ and $y$ are variables, and $R(x_1,\ldots, x_k)$, where \mbox{$R\in \sigma$}
is a $k$-ary relation symbol
and $x_1,\ldots, x_k$ are variables, by closing under the Boolean
connec\-tives~$\neg$,~$\wedge$~and~$\vee$, and by existential and
universal quantification~$\exists x$ and $\forall x$.
%
%Let $\sigma$ be a fixed vocabulary. \textit{Atomic formulae} are expressions of the form $x=y$ or $R(x_1,\ldots,x_n)$ for a $n$-ary relational symbol $R\in\sigma$ and variables $x,y,x_1,\ldots,x_n$. A \textit{term} is a variable, a constant symbol $c\in\sigma$ or a function symbol $f\in\sigma$ applied on other terms.
%
%\textit{First-order formulae} can be built from atomic formulae and are closed under boolean connectives ($\neg, \wedge, \vee$) and existential and universal quantifications ($\exists, \forall$):
%If $t_1,\ldots,t_n$ are terms, then $t_1=t_2$ and $R(t_1,\ldots,t_n)$ are formulae for a $n$-ary relation symbol $R\in\sigma$. If $\varphi$ and $\psi$ are formulae, then $\neg \varphi$, $\varphi\wedge\psi$ and $\varphi\vee\psi$ are formulae, as well as $\exists x_1,\ldots,x_n\varphi(x_1,\ldots,x_n)$ and $\forall x_1,\ldots,x_n\varphi(x_1,\ldots,x_n)$ for some variables $x_1,\ldots,x_n$.
%
A variable~$x$ not in the scope of a quantifier is a {\em free variable}. A formula without free variables is a {\em sentence}.
The {\em quantifier rank} $\mathrm{qr}(\varphi)$ of a formula $\varphi$ is the
maximum nesting depth of quantifiers in~$\phi$. We write $\FO_\sigma[q]$ for the set of all $\FO$ $\sigma$-formulas of quantifier rank at most~$q$,
or simply~$\FO[q]$ if $\sigma$ is clear from the context.
A formula without quantifiers is called {\em quantifier-free}.
%\begin{itemize}
%\item If $\varphi$ is atomic, then $\mathrm{qr}(\varphi)=0$.
%\item $\mathrm{qr}(\varphi_1\vee\varphi_2)=\mathrm{qr}(\varphi_1\wedge\varphi_2)=\max(\mathrm{qr}(\varphi_1), \mathrm{qr}(\varphi_2))$
%\item $\mathrm{qr}(\neg \varphi)=\mathrm{qr}(\varphi)$
%\item $\mathrm{qr}(\exists x \varphi)=\mathrm{qr}(\forall x\varphi)=\mathrm{qr}(\varphi)+1$
%\end{itemize}

If $\strA$ is a $\sigma$-structure
with universe $A$, then an {\em assignment} of the variables in~$\strA$
is a mapping $\bar a:\textsc{Var} \rightarrow A$. We use the standard
notation $(\strA, \bar a)\models \phi(\bar x)$ or $\strA \models \phi(\bar a)$
to indicate that $\phi$ is satisfied in $\strA$ when the free variables $\bar x$
of $\phi$ have been assigned by $\bar a$. We refer e.g.\ to the textbook~\cite{libkin2013elements}
for more background on first-order logic.

\section{Separator logic}
\label{sec:foconn}

In this section, we study the expressive power of separator logic
$\FOconn$. Formally, we assume that~$\sigma$ is
a signature that does not contain any of the relation symbols
$\conn_k$ for all~$k\geq 0$, and that it does contain a binary
relation symbol $E$, representing an edge relation. We
assume that $E$ is always
interpreted as an irreflexive and symmetric relation and
connectivity will always refer
to this relation. We let $\sigma+\conn\coloneqq \sigma\cup \{\conn_k : k\geq 0\}$, where each $\conn_k$ is a
$(k+2)$-ary relation symbol.

\begin{definition}\label{defn-foconn}
  The formulas of $(\FOconn)[\sigma]$ are the formulas of
  $\FO[\sigma+\conn]$. We usually simply write $\FOconn$,
  when $\sigma$ is understood from the context.

  For a $\sigma$-structure $\strA$, an assignment~$\bar a$
  and an $\FO+\conn$ formula $\phi(\bar x)$, we define the satisfaction relation $(\strA,\bar a)\models\phi(\bar x)$
  as for first-order logic, where an atomic predicate $\conn_k(x,y,
  z_1,\ldots, z_k)$ is evaluated as follows.
  Assume that the universe of $\strA$ is $A$ and let \mbox{$G=(A,E^\strA)$} be the
  graph on vertex set $A$ and edge set $E^\strA$. Then $(\strA, \bar a)$ models $\conn_k(x,y, z_1,\ldots, z_k)$ if and only if $\bar a(x)$ and $\bar a(y)$ are
  connected in $G-\{\bar a(z_1),\ldots, \bar a(z_k)\}$.
\end{definition}

 Note in particular
that if $\bar a(x)=\bar a(z_i)$ or $\bar a(y)=\bar a(z_i)$ for some $i\leq k$,
then \linebreak $(\strA,\bar a) \not\models \conn_k(x,y,z_1,\ldots, z_k)$.

\smallskip
We write $\FO+\conn_k$ for the fragment of $\FO+\conn$ that uses
only $\conn_\ell$ predicates for $\ell\leq k$. The quantifier rank of
an $\FO+\conn$ formula is defined as for plain first-order logic.
For structures $\strA$ with universe $A$ and $\bar a\in A^m$ and $\strB$ with
universe $B$ and $\bar b\in B^m$, we write $(\strA,\bar a)\equiv_{\conn}
(\strB, \bar b)$ if $(\strA,\bar a)$ and $(\strB,\bar b)$ satisfy the same
$\FO+\conn$ formulas, that is, for all $\phi(\bar x)$ we have $\strA\models
\phi(\bar a)\Leftrightarrow \strB\models \phi(\bar b)$. Similarly, we write
$(\strA,\bar a)\equiv_{\conn_k}(\strB, \bar b)$ and $(\strA,\bar a)\equiv_{\conn_{k,q}}
(\strB, \bar b)$ if $(\strA,\bar a)$ and $(\strB,\bar b)$ satisfy the same
$\FO+\conn_k$ formulas and the same $\FO+\conn_k$ formulas of quantifier
rank at most $q$, respectively.

\subsection{Expressive power of separator logic}

We now give examples of properties that are expressible with separator logic.

\begin{example}\label{exp:connectivity}
Connectivity is expressible in $\FO+\conn_0$ by the
formula \[\forall x\forall y\bigl(\conn_0(x,y)\bigr).\]

More generally,
for every non-negative integer $k$,
$(k+1)$-connectivity can be expressed by the formula
\[\forall x\forall y \forall  z_1\ldots \forall z_k \bigl(\bigwedge_{1\leq i\leq k} (x\neq z_i \wedge y\neq z_i) \rightarrow \conn_k(x,y,z_1,\ldots ,z_k)\bigr).\]
\end{example}

\pagebreak
\begin{example}
  We can express that there exists a cycle by
  \[\exists x\exists y \big(E(x,y)\wedge
  \exists z \big(\conn_1(z,x,y)\wedge
  \conn_1(z,y,x)\big)\big),\]
  hence, that a graph is acyclic by the negation of that formula. We write
  $\psi_{acyclic}$ for that formula.
  We can express that a graph is a tree by stating that it is connected and acyclic.
\end{example}

We can conveniently express deletion problems by relativizing formulas as follows.
For a formula $\phi$ that does not contain $z$ as a free variable
write $\del(z)[\phi]$ for the formula obtained from $\phi$ by
recursively replacing every subformula $\exists x \psi$ by $\exists x (x\neq z\wedge \psi)$, every subformula $\forall x \psi$ by $\forall x (x\neq z\rightarrow\psi)$
and every atomic formula $\conn_k(x,y,z_1,\ldots, z_k)$ by
$\conn_{k+1}(x,y,z_1,\ldots,z_k,z)$. Then
$(\strA,\bar a)\models \del(z)[\phi]$ if and only if
\mbox{$(\strA-\bar a(z), \bar a)\models \phi$}, where
$\strA-\bar a(z)$ denotes the substructure induced on the universe
of $\strA$ without $\bar a(z)$.

\begin{example}
We can state the existence of a feedback vertex set of size $k$ by
\[\exists z_1 \del(z_1)[\cdots [\exists z_k \del(z_k)[\psi_{acyclic}]\ldots].
\]

We can of course use the same principle to express any $\Pi$-deletion problem
that is $\FO+\conn$ expressible.
\end{example}

We can also, much more generally, express many elimination distance
problems.

\begin{example}
  The \emph{elimination distance} to a class $\Cc$ of graphs measures
  the number of recursive deletions of vertices needed for a graph $G$
  to become a member of $\Cc$. More precisely, a graph $G$ has
  elimination distance $0$ to $\Cc$ if $G \in \Cc$, and otherwise
  elimination distance at most~$k+1$ if in every connected component of $G$ we
  can delete a vertex such that the resulting graph has elimination
  distance at most $k$ to $\Cc$. Elimination distance was introduced by Bulian
  and Dawar~\cite{bulian2016graph} in their study of the parameterized
  complexity of the graph isomorphism problem and has
  recently obtained much attention in the
  literature, see e.g.~\cite{agrawal2021fpt,bulian2017parameterized,fomin2021parameterized,
  hols2019vertexcover,jansen2021vertex,LindermayrSV20}.

  Again, we define auxiliary notation. We write $\mathrm{comp}(x)$ for
  the connected component of (the valuation of) $x$. For a formula $\phi$ we write
  $\phi^{[\mathrm{comp}(x)]}$ for the formula obtained from $\phi$ by
  recursively replacing all subformulas $\exists y \psi$ by
  $\exists y (\conn_0(x,y) \wedge \psi)$ and all subformulas
  $\forall y \psi$ by
  $\forall y (\conn_0(x,y) \rightarrow \psi)$. Then
  $(\strA,\bar a)\models \phi^{[\mathrm{comp}(x)]}$ if and only if
  \mbox{$(\strA[\mathrm{comp}(\bar a(x))], \bar a)\models \phi$}, where
  $\strA[\mathrm{comp}(\bar a(x))]$ denotes the substructure induced on the
  connected component of $\bar a(x)$.

  Now assume $\Cc$ is a first-order definable class, say defined by a formula
  $\psi_\Cc$. Then elimination distance $0$ to $\Cc$ is defined by
  $\mathrm{ed}_0=\psi_\Cc$. If $\mathrm{ed}_k$ has been defined, then
  we can express elimination distance $k+1$ to $\Cc$ by the formula
  \[\mathrm{ed}_{k+1}\coloneqq \mathrm{ed}_k \vee \forall x\big(\exists y\ \del(y)[\mathrm{ed}_k]\big)^{[\mathrm{comp}(x)]}.\]
\end{example}

Our final example concerns the expressive power of separator
logic on finite words and finite trees. By the classical result of
B\"uchi, a language on words is regular if and only if it is
definable in $\MSO$. Here, words are
represented as finite structures over the vocabulary of the
successor relation and unary predicates representing the
letters of the alphabet. When considering first-order logic
on strings, it makes a big difference whether one considers
word structures over the successor relation or over its
transitive closure, the order relation. Languages definable by
$\FO$ over the order relation are exactly the star-free languages (see e.g.~\cite[Theorem~7.26]{libkin2013elements}),
while languages definable by $\FO$ over the successor
relation are exactly the locally threshold testable
languages~\cite[Theorem 4.8]{DBLP:reference/hfl/Thomas97}.
Similarly, MSO on trees can define exactly the tree regular languages
(defined via tree automata, see \cite[Theorem~7.30]{libkin2013elements}), while FO can only define a proper subclass
of the regular tree languages when the ancestor-descendant or even
only the parent-child relation is present.
This background was also the motivation of Boja\'nczyk, who studied a
variant of star-free expressions for graphs and showed that
these expressions exactly correspond to separator logic~\cite{bojanczyk2021separator}. In our example,
we show that separator logic on rooted trees has exactly the same
expressive power as first-order logic in the presence of the
ancestor-descendant relation. Let us write $\FO[<]$ for the latter
logic. On the other hand, we treat a rooted tree as a
graph-theoretic tree with an additional unary predicate marking the
root. In the degenerate case, we treat a word as a path, where one
of the endpoints is marked by a unary predicate as the smallest
vertex (the beginning of the word).

\begin{example}
On rooted trees (and similarly on words) \mbox{$\FO+\conn$} collapses to
\mbox{$\FO+\conn_1$} and has exactly the same expressive power as
$\FO[<]$ over trees with the ancestor-descendant relation. We
show first that $\conn_k(x,y,z_1,\ldots,z_k)$ can be expressed
in $\FO[<]$. For this, we need to ensure that $x$ and $y$ are not
equal to any $z_i$ and that no $z_i$ lies
on the unique path between $x$ and $y$ in the tree. We can
define the vertices on the unique path between $x$ and~$y$ by
first defining the least common ancestor of $x$ and $y$ by
the formula\linebreak $\mathrm{lca}(x,y,z)=z\leq x \wedge z\leq y \wedge
\neg \exists z'(z<z' \wedge z'\leq x \wedge z'\leq y)$. If $z$ is
the least common ancestor of $x$ and $y$, it remains
to state that none of the $z_i$ lies either between $x$ and $z$
or between $y$ and $z$, which is done by the formula
$\exists z\big(\mathrm{lca}(x,y,z) \wedge \bigwedge_{1\leq i\leq k}
\neg (z \leq z_i\leq x \vee z \leq z_i\leq y)\big)$.

Conversely, we show that we can define with $\FO+\conn_1$
the ancestor-descendant relation in rooted trees. Assume the root
is marked by the unary symbol $R$. Then
$x<y$ is equivalent to $\exists r \big(R(r) \wedge \conn_1(x,r,y) \wedge
\neg \conn_1(y,r,x)\big)$.
\end{example}

\subsection{The limits of separator logic}

We now study the limits of separator logic and show that
planarity
cannot be expressed in $\FO+\conn$.
Slightly abusing notation let us also write $\FOconn_k$ for the properties that
are expressible in $\FOconn_k$.
We show that there is a strict hierarchy of expressiveness:
$\FO+\conn_0\subsetneq \FO+\conn_1\subsetneq \FO+\conn_2 \subsetneq \ldots$
These results are based on an adaptation of the standard
Ehrenfeucht-Fra\"iss\'e game (EF game), which is commonly used in the study
of the expressive power of first-order logic.

\medskip\noindent \textbf{Ehrenfeucht-Fra\"iss\'e Games.}
The Ehrenfeucht-Fra\"iss\'e game is played by two players called
\emph{Spoiler} and \emph{Duplicator}. Given two structures $\strA$ and
$\strB$, Spoiler's aim is to show that the structures can be
distinguished by first-order logic (with formulas of a given quantifier rank),
while Duplicator wants to prove the opposite. The $q$-round EF game proceeds in $q$~rounds,
where each round consists of the following two steps.

\vspace{-2.5mm}
\begin{enumerate}
\item Spoiler picks an element $a\in\mathfrak{A}$ or an element $b\in\mathfrak{B}$.
\item Duplicator responds by picking an element of the other structure,
that is, she picks a $b\in\mathfrak{B}$ if Spoiler chose $a\in \strA$, and
she picks an $a\in\mathfrak{A}$ if Spoiler chose $b\in \strB$.
\end{enumerate}
\vspace{-2mm}

After $q$ rounds, the game stops. Assume the players have chosen
$\bar a=a_1,\ldots, a_q$ and $\bar b=b_1,\ldots, b_q$. Then Duplicator {\em wins} if the
mapping $a_i\mapsto b_i$ for all $1\leq i\leq q$ is a
partial isomorphism of $\strA$ and $\strB$.
We write for short $\bar a\mapsto\bar b$ for this mapping. Otherwise,
Spoiler wins. We say that Duplicator {\em wins the $q$-round
EF game} on $\strA$ and $\strB$ if she can force a win
no matter how Spoiler plays. We then write $\strA\simeq_q \strB$.

\begin{theorem}[Ehrenfeucht-Fra\"iss\'e, see e.g.~{\cite[Theorem 3.18]{libkin2013elements}}]
  \label{t:ef-theorem}
  Let $\mathfrak{A}$ and $\mathfrak{B}$ be two $\sigma$-structures where $\sigma$ is purely relational. Then
  $\strA\equiv_q\strB$ if and only if $\strA\simeq_q \strB$.
\end{theorem}

The EF game for $\FO$ naturally extends to separator logic. The \emph{($\conn_{k,q}$)-game} is played just as the $q$-round EF game, but the winning condition
is changed as follows. If in $q$ rounds the players have chosen
$\bar a=a_1,\ldots, a_q$ and $\bar b=b_1,\ldots, b_q$, then Duplicator
wins if

\vspace{-2.5mm}
\begin{enumerate}
  \item the mapping $\bar a\mapsto \bar b$ is a partial isomorphism of
  $\strA$ and $\strB$, and
  \item for every $\ell\le k$ and every sequence $(i_1,\ldots,i_{\ell+2})$ of numbers
  in $\{1,\ldots,q\}$ we have
  \[\strA\models \conn_{\ell}(a_{i_1},\ldots,a_{i_{\ell+2}}) \quad \Longleftrightarrow\quad \strB\models \conn_{\ell}(b_{i_1},\ldots,b_{i_{\ell+2}}).\]
\end{enumerate}

Otherwise,
Spoiler wins. We say that Duplicator wins the ($\conn_{k,q}$)-game on
$\strA$ and $\strB$ if she can force a win
no matter how Spoiler plays. We then write $\strA\simeq_{\conn_{k,q}} \strB$.

By following the lines of the proof of the classical Ehrenfeucht-Fra\"iss\'e Theorem
we can prove the following theorem.

\begin{theorem}
  \label{t:ef-conn-theorem}
  Let $\mathfrak{A}$ and $\mathfrak{B}$ be two $\sigma$-structures where $\sigma$ is purely rational (and contains a binary relation symbol $E$ that is
  interpreted on both structures as an irreflexive and symmetric relation). Then
  $\strA\equiv_{\conn_{k,q}}\strB$ if and only if $\strA\simeq_{\conn_{k,q}} \strB$.
\end{theorem}

The next theorem exemplifies the use of the ($\conn_{k,q}$)-game.

\begin{theorem}\label{thm-conn-planar}
  Planarity is not expressible in $\FOconn$.
\end{theorem}

\begin{figure}[ht]
  \centering
  \begin{subfigure}[b]{0.45\linewidth}
    \centering
    \begin{tikzpicture}
      \node (v11) at (0,0) [circle,draw] {$v_{1,1}$};
      \node (v31) at (2,0) [circle,draw] {$v_{2,1}$};
      \node (v12) at (0,-2) [circle,draw] {$v_{1,2}$};
      \node (v32) at (2,-2) [circle,draw] {$v_{2,2}$};
      \node (v13) at (0,-5) [circle,draw] {$v_{1,n}$};
      \node (v33) at (2,-5) [circle,draw] {$v_{2,n}$};

      \node (c0) at (0,.6) [blue] {$g_{-3}$};
      \node (c1) at (2,.6) [blue] {$g_{-2}$};
      \node (c2) at (0,-5.65) [blue] {$g_{-1}$};
      \node (c3) at (2,-5.65) [blue] {$g_0$};

      \path
      (v11) edge (v31)
      (v12) edge (v32)
      (v13) edge (v33)
      (v11) edge (v12)
      (v31) edge (v32)
      (v11) edge [bend right=40] (v13)
      (v31) edge [bend left=40] (v33);

      \path[dotted] (v12) edge (v13);
      \path[dotted] (v32) edge (v33);
    \end{tikzpicture}
    \subcaption{$G_q$}
  \end{subfigure}
  \begin{subfigure}[b]{0.45\linewidth}
    \centering
    \begin{tikzpicture}
      \node (v11) at (0,0) [circle,draw] {$v'_{1,1}$};
      \node (v31) at (2,0) [circle,draw] {$v'_{2,1}$};
      \node (v12) at (0,-2) [circle,draw] {$v'_{1,2}$};
      \node (v32) at (2,-2) [circle,draw] {$v'_{2,2}$};
      \node (v13) at (0,-5) [circle,draw] {$v'_{1,n}$};
      \node (v33) at (2,-5) [circle,draw] {$v'_{2,n}$};

      \node (c0) at (0,.65) [blue] {$h_{-3}$};
      \node (c1) at (2,.65) [blue] {$h_{-2}$};
      \node (c2) at (2,-5.7) [blue] {$h_{-1}$};
      \node (c3) at (0,-5.7) [blue] {$h_0$};

      \path
      (v11) edge (v31)
      (v12) edge (v32)
      (v13) edge (v33)
      (v11) edge (v12)
      (v31) edge (v32);

      \path[dotted] (v12) edge (v13);
      \path[dotted] (v32) edge (v33);

      \draw (v11) to [out=40,in=110] ($(v31)+(1,.1)$) to [out=-70,in=60] (v33);
      \draw (v31) to [out=140,in=70] ($(v11)+(-1,.1)$) to [out=-110,in=120] (v13);
    \end{tikzpicture}
    \subcaption{$H_q$}
  \end{subfigure}
  \caption{Planarity is not expressible in $\FOconn$}
  \label{fig:planar2}
\end{figure}

\begin{proof}
Assume planarity is expressible by a sentence $\phi$ of $\FO+\conn_k$ of
quantifier rank $q$. Without loss of generality, we may assume that
$k\leq q$, as otherwise, we have repetitions in the $\conn_k$
predicates that can be avoided by using $\conn_\ell$ predicates for
$\ell<k$.
  Let~$G_q$ and~$H_q$ be defined as shown in \cref{fig:planar2},
  where $n=2^{q+1}$. Then, $G_q$ is planar but $H_q$ embeds only
  in a surface of genus one (into the Möbius strip, which cannot be embedded into the plane).
%  To see this more clearly, looks at the following branch sets: $A_1 := \{v'_{1,1}\}$, $A_2 := \{v'_{2,2},\ldots,v'_{2,n-1}\}$, $A_3 := \{v'_{1,n}\}$, and
%  $B_1 := \{v'_{2,1}\}$, $B_2 := \{v'_{1,2},\ldots,v'_{1,n-1}\}$, $B_3 := \{v'_{2,n}\}$. Then for every $i,j$, there is a connection between $A_i$ and $B_j$.
We show that $G_q\simeq_{\conn_{k,q}} H_q$, contradicting the
assumption that $\phi$ must distinguish~$G_q$ and~$H_q$.
In fact, we prove an even stronger statement by giving Spoiler four
free moves $g_{-3}=v_{1,1}$, $g_{-2}=v_{2,1}$, $g_{-1}=v_{1,n}$
and $g_0=v_{2,n}$ in $G_q$ and forcing
Duplicator to respond with the vertices $h_{-3}=v'_{1,1}$, $h_{-2}=v'_{2,1}$, $h_{-1}=v'_{2,n}$ and $h_0=v'_{1,n}$ in $H_q$.
Note the twist in the last two vertices.
These extra moves are helpful to define Duplicator's winning strategy.

We define the $x$-distance of two nodes $v_{i,j}$ and $v_{k,\ell}$ as $d_x(v_{i,j},v_{k,\ell})=|i-k|$ and the $y$-distance as $d_y(v_{i,j},v_{k,\ell})=|j-\ell|$. Note that the $y$-distance is not the distance in the graphs, e.g.\ $d_y(g_{-3},g_{-1})=2^{q+1}-1$,
even though $g_{-3}$ and $g_{-1}$ are adjacent in~$G_q$.

Assume now that the first $i$ moves have been made in the game and
the players have selected the vertices $\bar g=(g_{-3},\ldots,g_0,g_1,\ldots,g_i)$
in $G_q$
(where $g_1,\ldots, g_i$ were freely chosen by the players), and $\bar h =(h_{-3},\ldots,h_0,h_1,\ldots,h_i)$ in $H_q$ (where $h_1,\ldots,h_i$ were freely
chosen by the players).
We prove by induction that Duplicator can play in such a way that after
round $i$ of the $(\conn_{k,q}$)-game the following conditions hold for all $-3\leq j,\ell\leq i$:

\vspace{-2.5mm}
  \begin{enumerate}
    \item if $g_j=v_{x,y}$, then $h_j=v'_{x',y}$, that is, corresponding
    pebbles are in the same row, and in particular $d_y(g_j,g_\ell)=d_y(h_j,h_\ell)$, and
    \item if $d_y(g_j,g_\ell)\leq 2^{q-i}$, then $d_x(g_j,g_\ell)=d_x(h_j,h_\ell)$.
  \end{enumerate}
  \vspace{-1mm}

These conditions together with the first four extra moves imply that the mapping $\bar g\mapsto \bar h$ is a partial
isomorphism of $G_q$ and $H_q$. Let us show that also
for every $0\leq \ell\le k$ and every sequence $(i_1,\ldots,i_{\ell+2})$ of numbers
in $\{-3,\ldots,i\}$ we have
$G_q\models \conn_{\ell}(g_{i_1},\ldots,g_{i_{\ell+2}})$ if and only if $H_q\models \conn_{\ell}(h_{i_1},\ldots,h_{i_{\ell+2}})$. Assume $G_q\models \conn_{\ell}(g_{i_1},\ldots,g_{i_{\ell+2}})$,
that is, $g_{i_1}$ and $g_{i_2}$ are connected after the deletion of
$g_{i_3},\ldots, g_{i_{\ell+2}}$, say by a path $P=v_{x_1,y_1}\ldots
v_{x_m,y_m}$, where $v_{x_1,y_1}=g_{i_1}$ and $v_{x_m,y_m}=g_{i_2}$.
Then there are no $g_{i_{j_1}}=v_{x,y}$ and $g_{i_{j_2}}=v_{x',y'}$
(for $j_1,j_2\geq 3$) with $y=y'=y_i$ and $x\neq x'$
for some $2\leq i\leq m-1$
(this would block a row along which the path goes,
which is not possible) and no
$g_{i_{j_1}}=v_{x,y}$ and $g_{i_{j_2}}=v_{x',y'}$
(for $j_1,j_2\geq 3$) with $y_i=y=y'-1=y_{i+1}-1$ and $x\neq x'$
for some $2\leq i\leq m-1$
(this would block a ``diagonal'' of which the path contains at least
one vertex, which
is not possible). By the first condition of the invariant
there are no $h_{i_{j_1}}=v_{x,y}$ and $h_{i_{j_2}}=v_{x',y'}$
(for $j_1,j_2\geq 3$) with $y=y'=y_i$ and $x\neq x'$
for some $2\leq i\leq m-1$ and by the second condition of
the invariant there are no
$h_{i_{j_1}}=v_{x,y}$ and $h_{i_{j_2}}=v_{x',y'}$
(for $j_1,j_2\geq 3$) with $y_i=y=y'-1=y_{i+1}-1$ and $x\neq x'$
for some $2\leq i\leq m-1$. Now, if $P'=v'_{x_1,y_1}\ldots
v'_{x_m,y_m}$ is not a path from $h_{i_1}$ to $h_{i_2}$ after
the deletion of $h_{i_3},\ldots, g_{i_{\ell+2}}$, it is possible to
reroute the path by switching the row appropriately, as the
$h_{i_j}$ never block a complete row or a diagonal, as shown
above. The case $H_q\models \conn_{\ell}(h_{i_1},\ldots,h_{i_{\ell+2}})$
is symmetrical.

  We now show that Duplicator can maintain this invariant throughout the game.
  For the initial configuration $i=0$, the  conditions are obviously
  fulfilled for $-3\leq j,\ell\leq 0$. Corresponding pebbles are in the
  same row and note that $d_y(g_{j},g_{\ell})=2^{q+1}-1$, for\linebreak $j\in \{-3,-2\}$
  and $\ell\in \{-1,0\}$ and analogously for $h_j$ and $h_\ell$.

  For the induction step, suppose that the conditions are fulfilled so far and that Spoiler is making his $(i+1)$-move in $G_q$ (the case of $H_q$ is symmetrical). We may assume that Spoiler does not choose a
  vertex that was chosen before, say Spoiler picks $g_{i+1}=v_{\_,a}$.
Duplicator must choose $h_{i+1}=v'_{\_,a}$ with the same
$y$-coordinate. We have to make sure that she can choose
the vertex with that $y$-coordinate satisfying the second condition.
Let $g_j=v_{\_,b}$ and $g_\ell=v_{\_,c}$ with $-3\leq j,\ell\leq i$ be
such that $b\leq a\leq c$ and there is no other $g_k=v_{\_,d}$ with
$b<d<c$. Intuitively, $g_j$ is the lowest pebble that was placed above
(or in the same row \mbox{as)~$g_{i+1}$}, while $g_k$ is the highest pebble that was placed below (or in the same row as) $g_{i+1}$.
  There are two cases:
  \begin{enumerate}
    \item $d_y(g_j,g_\ell)\leq 2^{q-i}$: Then by hypothesis, $d_x(h_j,h_\ell)=d_x(g_j,g_\ell)$ and $d_y(h_j,h_\ell)=d_x(g_j,g_\ell)$. Here, Duplicator chooses the unique
    $h_{i+1}=v'_{\_,a}$ such that $d_x(h_j,h_{i+1})=d_x(g_j,g_{i+1})$, and we have $d_x(h_\ell,h_{i+1})=d_x(g_\ell,g_{i+1})$.
    \item $d_y(g_j,g_\ell)> 2^{q-i}$: Then $d_y(h_j,h_\ell)> 2^{q-i}$ and there are three possibilities:
    \begin{itemize}
      \item $d_y(g_j,g_{i+1})\leq 2^{q-(i+1)}$: Then $d_y(g_\ell,g_{i+1})> 2^{q-(i+1)}$, and Duplicator chooses\linebreak $h_{i+1}=v'_{\_,a}$ such that $d_x(h_j,h_{i+1})=d_x(g_j,g_{i+1})$.
      Hence, $d_y(h_\ell,h_{i+1})> 2^{q-(i+1)}$.
      \item $d_y(g_\ell,g_{i+1})\leq 2^{q-(i+1)}$: Then $d_y(g_j,g_{i+1})> 2^{q-(i+1)}$. Similarly to the previous case, Duplicator chooses $h_{i+1}=v'_{\_,a}$ such that $d_x(h_\ell,h_{i+1})=d_x(g_\ell,g_{i+1})$. Consequently, $d_y(h_j,h_{i+1})> 2^{q-(i+1)}$.
      \item $d_y(g_j,g_{i+1})> 2^{q-(i+1)}$ and $d_y(g_\ell,g_{i+1})> 2^{q-(i+1)}$: Here, Duplicator can choose $h_{i+1}=v'_{1,a}$ or $h_{i+1}=v'_{2,a}$ as she wants. We get that $d_y(h_j,h_{i+1})\geq 2^{q-(i+1)}$ and $d_y(h_\ell,h_{i+1})\geq 2^{q-(i+1)}$.
    \end{itemize}
  \end{enumerate}

  Thus, in all cases, the conditions are fulfilled and Duplicator wins the ($\conn_{k,q}$)-game on $G_q$ and $H_q$. Hence, planarity is not definable in $\FOconn$.
\end{proof}

As a graph is planar if and only if it excludes $K_5$ and $K_{3,3}$ as
(topological) minors and we will show that this can be expressed using disjoint paths
predicates, we conclude that the disjoint paths predicate cannot be
expressed with $\FO+\conn$.

\begin{corollary}\label{crl:disjp-no-foconn}
The disjoint paths problem cannot be expressed in $\FO+\conn$.
\end{corollary}

The proof of the next theorem is deferred to the next section, as it is
a consequence of the fact that the even stronger logic
$\FODP$ cannot express bipartiteness (\cref{thm-dp-bipart}).

\begin{theorem}\label{thm-con-bipart}
  Bipartiteness cannot be expressed in $\FOconn$.
\end{theorem}

Finally, we show that the $\FO+\conn_k$ hierarchy is strict by proving
that $(k+2)$-connectivity cannot be expressed by $\FO+\conn_k$.
On the other hand, $(k+2)$-connectivity can be expressed by $\FO+\conn_{k+1}$ (\cref{exp:connectivity}).

\begin{theorem}\label{thm-strict-hierarchy}
  $(k+2)$-connectivity cannot be expressed by $\FO+\conn_k$. In
  particular, the $\FO+\conn_k$ hierarchy is strict, that is, $\FO+\conn_0\subsetneq
  \FO+\conn_1\subsetneq \ldots$
\end{theorem}

\begin{proof}
  Let $k$ be an integer. For every integer $q$, we choose two graphs $G_q$ and $H_q$ such that:
  \begin{itemize}
    \item $G_q$ is connected,
    \item $H_q$ is not connected, and
    \item $G_q \simeq_q H_q$.
  \end{itemize}

This is possible, as connectivity is not first-order definable and
$\simeq_q$ has only finitely many equivalence classes.

  Then, we define the graph $G_q^k$ (resp.\ $H_q^k$) as the disjoint union of $G_q$ (resp.\ $H_q$) and $K_{k+1}$, a clique of size $k+1$,
  and connect the vertices of the clique with all vertices of $G_q$ (resp.~$H_q$), that
  is, we add the additional edges such that $(x,y)\in E(G_q^k)$ (resp.\ $(x,y)\in E(H_q^k)$) if $x\in G_q$ (resp.\ $x\in H_q$) and $y\in K_{k+1}$.
  Obviously, $G_q^k$ is ($k+2$)-connected (the deletion of any
  $k+1$ vertices cannot disconnect $G_q^k$), while $H_q^k$
  is not ($k+2$)-connected (the deletion of the copy of $K_{k+1}$
  disconnects $H_q^k$).

  The same argument shows that every $\conn_k(x,y,z_1,\ldots, z_k)$
  can be expressed by an atomic plain first-order formula:
  in both graphs (the valuations of) $x$ and $y$ are not connected after
  the deletion of (the valuations of) $z_1,\ldots, z_k$ if and only if
  $x$ or $y$ is equal to one of the~$z_i$. Hence, to prove
  $G_q^k\simeq_{\conn_{k,q}}H_q^k$ it suffices to prove
  $G_q^k\simeq_q H_q^k$, and this finishes the proof.
  \pagebreak
  \begin{claim}
    For all integers $q,k$ we have $G_q^{k}\simeq_q H_q^{k}$.
  \end{claim}

  \begin{proof}
    The following is obviously a winning strategy for Duplicator in the
    $q$-round EF game on $G_q^k$ and $H_q^k$.
    If Spoiler plays a pebble in the subgraph $G_q$ or $H_q$, Duplicator can respond by a pebble in the subgraph $H_q$ or $G_q$ according to the winning strategy of Duplicator in the EF game on $G_q$ and $H_q$.
    Otherwise, if Spoiler picks a pebble in the subgraph $K_{k+1}$ of $G_q^k$ or~$H_q^k$, Duplicator can respond by a pebble in the subgraph $K_{k+1}$ of the other graph $H_q^k$ or $G_q^k$.
  \end{proof}
  This concludes the proof of \cref{thm-strict-hierarchy}.
\end{proof}

% !TEX root = main.tex

\section{Disjoint-paths logic}
\label{sec:fodp}

In this section, we study the expressive power of disjoint-paths logic
$\FODP$. We again fix a signature $\sigma$ that does not contain
the symbol $\disjp_k$ for any $k\geq 1$ and that does contain
a binary (edge) relation symbol $E$. The disjoint paths
predicates will always refer to this relation.
We let $\sigma+\disjp\coloneqq \sigma\cup \{\disjp_k : k\geq 1\}$, where each $\disjp_k$ is a
$2k$-ary relation symbol.

\begin{definition}\label{def:fodp}
  The formulas of $(\FODP)[\sigma]$ are the formulas of
  $\FO[\sigma+\disjp]$. We usually simply write $\FODP$,
  when $\sigma$ is understood from the context.

  For a $\sigma$-structure $\strA$, an assignment~$\bar a$
  and an $\FODP$ formula $\phi(\bar x)$, we define the satisfaction relation $(\strA,\bar a)\models\phi(\bar x)$
  as for first-order logic, where an atomic predicate $\disjp_k[(x_1, y_1),\ldots
  (x_k,y_k)]$ is evaluated as follows.
  Assume that the universe of~$\strA$ is~$A$ and let \mbox{$G=(A,E^\strA)$} be the
  graph on vertex set $A$ and edge set $E^\strA$. Then $(\strA, \bar a)$
  models $\disjp_k[(x_1,y_1),\ldots,(x_k, y_k)]$ if and only if in $G$ there
  exist $k$ internally vertex-disjoint paths $P_1,\ldots, P_k$, where
  $P_i$ connects~$\bar a(x_i)$ and~$\bar a(y_i)$.
\end{definition}

As previously mentioned, it is natural to consider these predicates for both
undirected and directed graphs. We will, however, in this work only study the undirected case.

\smallskip
We write $\FODP_k$ for the fragment of $\FODP$ that uses
only $\disjp_\ell$ predicates for $\ell\leq k$.
The quantifier rank of
an $\FODP$ formula is defined as for plain first-order logic.
For structures $\strA$ with universe $A$ and $\bar a\in A^m$ and $\strB$ with
universe $B$ and $\bar b\in B^m$, we write $(\strA,\bar a)\equiv_{\DP}
(\strB, \bar b)$ if $(\strA,\bar a)$ and $(\strB,\bar b)$ satisfy the same
$\FODP$ formulas, that is, for all $\phi(\bar x)$ we have $\strA\models
\phi(\bar a)\Leftrightarrow \strB\models \phi(\bar b)$. Similarly, we write
$(\strA,\bar a)\equiv_{\DP_k}(\strB, \bar b)$ and $(\strA,\bar a)\equiv_{\DP_{k,q}}
(\strB, \bar b)$ if $(\strA,\bar a)$ and $(\strB,\bar b)$ satisfy the same
$\FODP_k$ formulas and the same $\FODP_k$ formulas of quantifier
rank at most $q$, respectively.

\subsection{Expressive power of disjoint-paths logic}

We now study the expressive power of disjoint-paths logic.

\begin{observation}
  $\FOconn \subseteq \FODP$ because
  $\conn_k(x,y,z_1,\ldots,z_k)$ is equivalent to $\disjp_{k+1}[(x,y),(z_1,z_1),\ldots ,(z_k,z_k)] \wedge \bigwedge\limits_{i\le k}(z_i\neq x \wedge z_i \neq y) $.
\end{observation}

Moreover, the inclusion is strict because planarity is not expressible in $\FOconn$ as seen in \cref{crl:disjp-no-foconn}.
We show that planarity and in fact the
property that a graph contains a fixed (topological) minor can be expressed
in $\FODP$.

\begin{example}
  For every fixed graph $H$, there is an
  $\FODP$ formula $\phi^{top}_H$ such that $G\models\phi^{top}_H$ if and only if
  $H\minor^{top} G$.

  Let $n,m,\ell$ respectively be the number of vertices, edges, and isolated vertices in $H$.
  Let $x_1,\ldots x_n$ be $n$ variables. Let $e_1,\ldots,e_m$ be the list of edges
  of $H$, and let $v_{j_s}$ and $v_{j_t}$ be the two endpoints of $e_j$. Finally,
  let $v_{i_1},\ldots,v_{i_\ell}$ be the isolated vertices of $H$. Then,
  \begin{align*}
  \phi^{top}_H :=\exists x_1,\ldots  x_n \big(& \bigwedge_{i\neq j} x_i\neq x_j \quad \wedge \\ & \disjp[(x_{e_{1_s}},x_{e_{1_t}}), \ldots (x_{e_{m_s}},x_{e_{m_t}}),(x_{i_1},x_{i_1}),\ldots (x_{i_\ell},x_{i_\ell})] \big).
  \end{align*}
\end{example}

\begin{example}
  For every fixed graph $H$, there is an $\FODP$ formula $\phi_H$ such that $G\models\phi_H$ if and only if $H\minor G$.
  This is because, for every graph $H$, there exists a finite family of graphs $H_1,\ldots,H_\ell$ such that $H\minor G$ if and only if there is an $i\le \ell$ such that $H_i\minor^{top}G$. This family can be obtained by considering all possibilities of replacing every branch set representing a vertex of $H$ of degree $d\geq 3$
with a tree with at most $d$ leaves and hardcoding their shapes by disjoint paths.
\end{example}

\begin{example}
  Planarity can be expressed in $\FODP$. This is a corollary of the previous example, using the formula $\phi_{planar} := \neg \phi_{K_5} \wedge \neg \phi_{K_{3,3}}$.
\end{example}

\subsection{The limits of disjoint-paths logic}

We now study the limits of disjoint-paths logic and show
that bipartiteness cannot be expressed in $\FODP$.
We also show that the hierarchy on $(\FODP_k)_{k\geq 1}$ is strict.
These results are based again on an adaptation of the standard Ehrenfeucht-Fra\"iss\'e game.

The \emph{($\DP_{k,q}$)-game} is played just as the $q$-round EF game, but the winning condition is changed as follows. If in $q$ rounds the players have chosen $\bar a=a_1,\ldots, a_q$ and $\bar b=b_1,\ldots, b_q$, then Duplicator wins if

\begin{enumerate}
  \item the mapping $\bar a\mapsto \bar b$ is a partial isomorphism of $\strA$ and $\strB$, and
  \item for every $\ell\le k$ and every sequence $(i_1,\ldots,i_{2\ell})$ of numbers in $\{1,\ldots,q\}$ we have
    \begin{align*}
      \strA&\models \disjp[(a_{i_1},a_{i_2}),\ldots,(a_{i_{2\ell-1}},a_{i_{2\ell}})]\\
      \quad \Longleftrightarrow\quad \strB&\models \disjp[(b_{i_1},b_{i_2}),\ldots,(b_{i_{2\ell-1}},b_{i_{2\ell}})].
    \end{align*}
\end{enumerate}

Otherwise, Spoiler wins. We say that Duplicator {\em wins the ($\DP_{k,q}$)-game} on
$\strA$ and $\strB$ if she can force a win
no matter how Spoiler plays. We then write $\strA\simeq_{\DP_{k,q}} \strB$.

By following the lines of the proof of the classical Ehrenfeucht-Fra\"iss\'e Theorem
we can prove the following theorem.

\begin{theorem}
  \label{t:ef-dp-theorem}
  Let $\mathfrak{A}$ and $\mathfrak{B}$ be two $\sigma$-structures where $\sigma$ is purely rational (and contains a binary relation symbol $E$ that is
  interpreted on both structures as an irreflexive and symmetric relation). Then
  $\strA\equiv_{\DP_{k,q}}\strB$ if and only if $\strA\simeq_{\DP_{k,q}} \strB$.
\end{theorem}

\begin{theorem}\label{thm-dp-bipart}
  Bipartiteness is not definable in $\FODP$.
\end{theorem}

\begin{proof}
  Let $q$ be an integer, and
  let $G$ be a cycle graph with $2^q$ vertices and $H$ a cycle graph with $2^q+1$ vertices.
  Then, $G$ is bipartite because it has an even number of vertices, and $H$ is not bipartite because it has an odd number of vertices.
  We want to show that $G\simeq_{\DP_{k,q}} H$ by induction over $q$.

  We define the distance $d(x,y)$ of two vertices $x$ and $y$ as the length of the shortest path between $x$ and $y$.

  Let $\bar g=(g_1,\ldots,g_i)$ be the first $i$ moves in $G$ and similarly $\bar h=(h_1,\ldots,h_i)$ the first $i$ moves in $H$. We can prove by induction that Duplicator can play in such a way that after round $i$ of the ($\DP_{k,q}$)-game the following conditions hold for all $j,\ell\leq i$:
  \begin{enumerate}
    \item If $d(g_j,g_\ell)<2^{q-i+1}$, then $d(g_j,g_\ell)=d(h_j,h_\ell)$.
    \item If $d(g_j,g_\ell)\geq 2^{q-i+1}$, then $d(h_j,h_\ell)\geq 2^{q-i+1}$.
    \item The pebbles are placed in $G$ and $H$ with the same ``circular order''.
  \end{enumerate}
  By the first two conditions, the partial isomorphism $\bar g\mapsto\bar h$ can be ensured.
  Furthermore, the third condition implies that the second condition for
  Duplicator's win is also satisfied.

  The base case $i=1$ of the induction is trivial because $d(g_1,g_1)=d(h_1,h_1)=0$.

  For the induction step, suppose that $G\simeq_{\DP_{k,i}} H$ holds and Spoiler is making his $(i+1)$-st move in G. The case of $H$ is equivalent.

  If Spoiler picks $g_j$ for some $j\leq i$, a pebble that was already played before, Duplicator can choose $h_j$, and the conditions are fulfilled by the induction hypothesis.
  Otherwise, Spoiler picks a pebble $g_{i+1}$ that wasn't played before. Now we have to differentiate two cases:
  \begin{enumerate}
    \item There is only one other pebble that was already played, $g_j=g_1,j\leq i$. Then, we can find $h_{i+1}$ such that $d(h_1,h_{i+1})=d(g_1,g_{i+1})$.
    \item $g_{i+1}$ lies on the shortest path of $g_j$ and $g_\ell$ with $j,\ell\leq i$ such that there is no other $g_n,n\leq i$ that lies on this path. Then, there are two possibilities:
    \begin{itemize}
      \item $d(g_j,g_\ell)< 2^{q-i+1}$: Then $d(h_j,h_\ell)<2^{q-i+1}$ and we can find $h_{i+1}$ on the shortest path of $h_j$ and $h_\ell$ such that $d(h_j,h_{i+1})=d(g_j,g_{i+1})$ and $d(h_{i+1},h_\ell)=d(g_{i+1},g_\ell)$.
      \item $d(g_j,g_\ell)\geq 2^{q-i+1}$: Then $d(h_j,h_\ell)\geq 2^{q-i+1}$ and there are three cases:
      \begin{enumerate}
        \item $d(g_j,g_{i+1})< 2^{q-i}$: Then $d(g_{i+1},g_\ell)\geq 2^{q-i}$ and we can choose $h_{i+1}$ on the shortest path of $h_j$ and $h_\ell$ such that $d(h_j,h_{i+1})=d(g_j,g_{i+1})$ and $d(h_{i+1},h_\ell)\geq 2^{q-i}$.
        \item $d(g_{i+1},g_\ell)< 2^{q-i}$: This case is similar to the previous one.
        \item $d(g_j,g_{i+1})\geq 2^{q-i}$ and $d(g_{i+1},g_\ell)\geq 2^{q-i}$: Since $d(h_j,h_\ell)\geq 2^{q-i+1}$, we can find $h_{i+1}$ with $d(h_j,h_{i+1})\geq 2^{q-i}$ and $d(h_{i+1},h_\ell)\geq 2^{q-i}$ in the middle of the shortest path of $h_j,$ and $h_\ell$.
      \end{enumerate}
    \end{itemize}
  \end{enumerate}
  Thus, in all cases, the conditions are fulfilled. This completes the inductive proof.
\end{proof}

We now show that the hierarchy on $(\FODP_k)_{k\geq 1}$ is strict.

\begin{lemma}\label{lem-dp-kconn}
  For all integers $k\geq 1$, $2k$-connectivity is not expressible in $\FODP_k$.
\end{lemma}

\begin{proof}
  Let $k$ be an integer. For every integer $q$, we define two graphs $G_q$ and $H_q$ such that:
  \begin{itemize}
    \item $G_q$ is $2$-connected,
    \item $H_q$ is $1$-connected but not $2$-connected, and
    \item $G_q \simeq_q H_q$
  \end{itemize}
  For example, take $G_q$ the cycle with $2^{q+1}$ many elements, together with an apex vertex, while $H_q$ is the disjoint union of two cycles with $2^q$ many elements each, together with an apex vertex (see~\cref{fig:fodp-hierarchy}).
  \begin{figure}[t]
    \centering
    \begin{subfigure}[b]{.45\linewidth}
      \centering
      \begin{tikzpicture}
        \node (v1) at (0,1.7) [circle,draw,scale=.7] {$v_1$};
        \node (v2) at (1.2,1.2) [circle,draw,scale=.7] {$v_2$};
        \node (v3) at (1.7,0) [circle,draw,scale=.7] {$v_3$};
        \node (v4) at (1.2,-1.2) [circle,draw,scale=.7] {$v_4$};
        \node (v5) at (0,-1.7) [circle,draw,scale=.7] {$v_5$};
        \node (v6) at (-1.2,-1.2) [circle,draw,scale=.7] {$v_6$};
        \node (v7) at (-1.7,0) [circle,draw,scale=.7] {$v_7$};
        \node (v8) at (-1.2,1.2) [circle,draw,scale=.45] {$v_{2^{q+1}}$};

        \node (k1) at (0,0) [circle,draw,scale=.7] {$v_0$};

        \path
        (v1) edge [bend left=15] (v2)
        (v2) edge [bend left=15] (v3)
        (v3) edge [bend left=15] (v4)
        (v4) edge [bend left=15] (v5)
        (v5) edge [bend left=15] (v6)
        (v6) edge [bend left=15] (v7)
        (v7) edge [dotted,bend left=15] (v8)
        (v8) edge [bend left=15] (v1);

        \path
        (k1) edge (v1)
        (k1) edge (v2)
        (k1) edge (v3)
        (k1) edge (v4)
        (k1) edge (v5)
        (k1) edge (v6)
        (k1) edge (v7)
        (k1) edge (v8);
      \end{tikzpicture}
      \subcaption{$G_q$}
    \end{subfigure}
    \begin{subfigure}[b]{.45\linewidth}
      \centering
      \begin{tikzpicture}
        \node (v1) at (0,1.1) [circle,draw,scale=.6] {$v'_1$};
        \node (v2) at (1.1,0) [circle,draw,scale=.6] {$v'_2$};
        \node (v3) at (0,-1.1) [circle,draw,scale=.6] {$v'_3$};
        \node (v4) at (-1.1,0) [circle, draw,scale=.5] {$v'_{2^q}$};

        \node (v5) at (3,1.1) [circle,draw,scale=.4] {$v'_{2^q+1}$};
        \node (v6) at (4.1,0) [circle,draw,scale=.4] {$v'_{2^q+2}$};
        \node (v7) at (3,-1.1) [circle,draw,scale=.4] {$v'_{2^q+3}$};
        \node (v8) at (1.9,0) [circle,draw,scale=.4] {$v'_{2^{q+1}}$};

        \node (k1) at (1.5,1.5) [circle,draw,scale=.6] {$v'_0$};

        \path
        (v1) edge [bend left=35] (v2)
        (v2) edge [bend left=35] (v3)
        (v3) edge [dotted,bend left=35] (v4)
        (v4) edge [bend left=35] (v1);

        \path
        (v5) edge [bend left=35] (v6)
        (v6) edge [bend left=35] (v7)
        (v7) edge [dotted, bend left=35] (v8)
        (v8) edge [bend left=35] (v5);

        \path
        (k1) edge (v1)
        (k1) edge (v2)
        (k1) edge (v3)
        (k1) edge (v4)
        (k1) edge (v5)
        (k1) edge (v6)
        (k1) edge (v7)
        (k1) edge (v8);
      \end{tikzpicture}
      \subcaption{$H_q$}
    \end{subfigure}
    \caption{$\FODP$ hierarchy is strict}
    \label{fig:fodp-hierarchy}
  \end{figure}

  We then define $G_q^k$ (resp. $H_q^k$) as the lexicographical product of $G_q$ (resp. $H_q$) with $K_{2k}$, the clique with $2k$ elements.
  More precisely, if $G_q =(V,E)$, where $V=\{1,\ldots,n\}$, then $G_q^k := (V',E')$ where:
  \begin{itemize}
    \item $V' := \{v_{1,1},\ldots,v_{1,2k},\ldots,v_{n,1},\ldots,v_{n,2k}\}$
    \item $E' := \{ \{v_{i,j},v_{i',j'}\} ~:~ i=i' \vee (i,i')\in E\}$.
  \end{itemize}
  One can view $G_q^k$ as $2k$ copies of $G_q$ on top of each other. Vertices are replaced by $2k$-cliques, and edges are replaced by $(2k,2k)$-bicliques.
  A direct consequence of the definition is the following equivalence.

  \begin{claim}\label{cla-fo-eq}
    For all integers $q,k$, we have that $G_q^k \simeq_q H_q^k$.
  \end{claim}
  \begin{proof}
    Duplicator's strategy follows the one derived from $G_q \simeq_q H_q$.
    If Spoiler picks a vertex $v_{i,j}\in G_q^k$, then Duplicator can respond by choosing the vertex $v_{i',j}\in H_q^k$ where $v_{i'}\in H_q$ is Duplicator's respond to $v_i\in G_q$.
  \end{proof}

  We then show that over $G_q^k$ and $H_q^k$, the predicate $\disjp_k[~]$ is always true and therefore that, for these structures, $(\FODP_k)[q]$ collapses to $\FO[q]$.

  \begin{claim}\label{cla-disjp-true}
    For every integers $q,k$, for every $k$-tuples $\bar a, \bar b$, we have that
    $G_q^k$ and $H_q^k$ both model $\disjp_k[(a_1,b_1),\ldots,(a_k,b_k)]$.
  \end{claim}
  \begin{proof}
    The proofs for $G_q^k$ and $H_q^k$ are identical, so we only do it for $G_q^k$. Remember that $n$ is the number of vertices in $G_q$.
    The idea is that each of the $k$ paths uses at most two ``copies'' of each
    vertex of $G_q$, hence $2k$ ``copies'' is enough for all paths to exists.
    For every $i\le n$, let $B_i:= \{v_{i,j} ~:~ j\le 2k \}$, and $F_i := \{v_{i,j} ~:~ j \le 2k \wedge v_{i,j}\not\in \bar a \wedge v_{i,j}\not \in \bar b\}$. We call $B_i$ the set of vertices in {\em position} $i$, and $F_i$ the {\em free vertices} in position $i$.
    We then compute each path, starting with $(a_1,b_1)$.

    Let $i,j,i',j'$ such that $a_1=v_{i,j}$ and $b_1=v_{i',j'}$. If $i = i'$, then there is nothing to do as $a_1$ and $b_1$ are neighbors. Otherwise, note that for every $i''\le n$, $F_{i''}\neq \emptyset$, because there are only $2k-2$ elements among $a_2,\ldots,a_k,b_2,\ldots,b_k$.
    Since $G_q$ is a connected graph, there is a path from $i$ to $i'$. For every inner node $i''$ of this path, we can select a vertex $v\in F_{i''}$.\linebreak
    We can therefore create a path in $G_q^k$ from $a_1$ to $b_1$ where all inner vertices are free vertices. We then remove these vertices from the sets of free vertices.

    Let now $1<\ell\le k$, and let $i,j,i',j'$ such that $a_\ell=v_{i,j}$ and $b_\ell=v_{i',j'}$. We assume that the first $\ell-1$ paths have already been computed.
    Observe that here again, if $i=i'$ there is nothing to do. Otherwise, we again have that for every $i''$, $F_{i''}$ is not empty. This is because for every $s\le k$, the path from $a_s$ to $b_s$ intersects $B_{i''}$ at most twice \textit{\small (at most once for the inner vertices, and twice when the two endpoints are both in position $i''$)}. Therefore, we can select a path in $G_q$ from $i$ to $i'$ and for each $i''$ in this path, pick a vertex $v\in F_{i''}$.
  \end{proof}

  With Claim~\ref{cla-disjp-true}, we can replace formulas of $(\FODP_k)[q]$ by formulas of $\FO[q]$. Thanks to Claim~\ref{cla-fo-eq}, $G_q^k \simeq_q H_q^k$, we conclude that $G_q^k \simeq_{\DP_{k,q}} H_q^k$. So $\FODP_k$ cannot express $2k$-connectivity.
  Note that this bound is tight for these structures i.e.~$G_q^k \not\simeq_{\DP_{k+1,q}} H_q^k$.
\end{proof}

\begin{lemma}
  The $\FODP_k$ hierarchy is strict, that is, $\FODP_1\subsetneq
  \FODP_2\subsetneq \ldots$
\end{lemma}
\begin{proof}
  Consider the structures in the proof of \cref{lem-dp-kconn}, which
  are indistinguishable in $\FODP_k$. The following
  sentence of $\FODP_{k+1}$ distinguishes $G_q^k$ and $H_q^k$:
  $$\exists a_1\ldots\exists b_{k+1}~ \neg \disjp_{k+1}[(a_1,b_1),\ldots, (a_{k+1},b_{k+1})] $$
  In $H_q^k$, pick $i$ such that $H_q\setminus {i}$ is not connected ($i'$ and $i''$ two disconnected vertices). Then pick $a_j = v_{i,j}$ if $j\le k$, $b_j = v_{i,k+j}$ if $j\le k$, and finally $a_{k+1}=v_{i',1}$, $b_{k+1}=v_{i'',1}$.
  Intuitively, this means that the vertices $v_{i,j}$ are ``blocked'' for every $j\leq 2k$ by the first $k$ paths and can therefore not be used for the $(k+1)$-st path such that this disjoint path does not exist.

  $G_q^k$ does not satisfy the formula because even if we ``block'' such a clique, there is still a disjoint path connecting every pair of vertices because $G_q$ is 2-connected.
\end{proof}

% !TEX root = main.tex

\section{Connection to other logics}
\label{sec:other-logics}

In this section, we compare the expressive power of the separator logic and the disjoint-paths logic with monadic second-order logic and transitive-closure logic.
\cref{fig:connection} depicts the connections between these logics.

\subsection{Monadic second-order logic}

Monadic second-order logic ($\MSO_1$) allows quantification over sets of vertices in addition to the first-order quantifiers.
It has a higher expressive power than first-order logic because for example connectivity is expressible in $\MSO_1$ and every first-order formula can be expressed with the first-order quantifiers.
Connectivity is expressible by
\[\forall R \Big( \big(\exists x R(x)\wedge\exists x\neg R(x)\big) \to\exists x\exists y \big( R(x)\wedge\neg R(y)\wedge E(x,y) \big) \Big)\]

By an extension of this formula, we can say that a given set $S$ is connected:
\begin{align*}
    \mathrm{conn}\text{-}\mathrm{set}(S) := \forall R \Big( \big(&R \subseteq S \wedge \exists x ~ R(x)\wedge\exists x ~ (S(x)\wedge\neg R(x)) \big)\\
    &\to\exists x\exists y\big(R(x)\wedge\neg R(y) \wedge S(y)\wedge E(x,y)\big) \Big)
\end{align*}

Furthermore, we can express the connectivity operators in $\MSO_1$. The connectivity operator $\conn_0(x,y)$ can be expressed by:
\[\conn_0(x,y):=\forall R\Big(R(x)\wedge\forall v\forall w \big( (R(v)\wedge E(v,w))\to R(w)\big)\to R(y)\Big)\]
and $\conn_k(x,y,z_1,\ldots,z_k)$ using $\mathrm{conn}\text{-}\mathrm{set}(S)$ by:
\[\conn_k(x,y,z_1,\ldots,z_k) := \exists S~ \big( \mathrm{conn}\text{-}\mathrm{set}(S) \wedge S(x) \wedge S(y) \wedge \bigwedge\limits_{i\le k} \neg S(z_i) \big).\]

We can express the disjoint paths predicates $\disjp_k[(x_1,y_1),\ldots,(x_k,y_k)]$ by:
\begin{align*}
    \exists S_1,\ldots,S_k \bigg( &\bigwedge\limits_{i\le k} \Big( S_i(x_i) \wedge S_i(y_i) \wedge \mathrm{conn}\text{-}\mathrm{set}(S_i) \Big)\\
    \wedge &\bigwedge\limits_{i<j\le k} \forall z \Big( \big(S_i(z) \wedge S_j(z)\big) \rightarrow \big((z=x_i\vee z=y_i) \wedge (z=x_j \vee z=y_j) \big) \Big)\bigg)
\end{align*}

Since the disjoint paths operators are expressible in $\MSO_1$, $\FODP$ is included in $\MSO_1$.
This inclusion is strict because it is well-known that bipartiteness is expressible in $\MSO_1$:
\begin{align*}
    \exists R_1\exists R_2\Big(\forall x\big(R_1(x)\leftrightarrow\neg R_2(x)\big)\wedge\bigwedge_{i\leq 2}\forall x\forall y\big((R_i(x)\wedge R_i(y))\to\neg E(x,y)\big)\Big)
\end{align*}

but we showed in \cref{thm-dp-bipart} that bipartiteness is not expressible in $\FODP$.

\subsection{Transitive-closure logic}

Transitive-closure logic $\TC^i_j$ is the enrichment of first-order logic with the transitive-closure operator $[\TC_{\bar x,\bar y}\varphi(\bar x,\bar y)]$ where $\bar x$ and $\bar y$ are tuples of length $i$ and $\varphi$ is a formula with at most~$j$ free variables other than $\bar x$ and $\bar y$.

Every $\FOconn_k$ formula can be expressed in $\TC^1_k$ because the $\Conn_k$ operator can be expressed with the help of the transitive-closure operator: \[\Conn_k(x,y,z_1,\ldots,z_k)=[\TC_{v,w} E(v,w)\wedge v\neq z_1\wedge\ldots\wedge v\neq z_k\wedge w\neq z_1\wedge\ldots\wedge w\neq z_k](x,y)\]
In fact, $\TC^1_k$ is more expressible than $\FOconn_k$, as it
can express bipartiteness~\cite[Example~7.2]{grohe2008logic}.
On the other hand, $2$-connectivity can naturally be expressed in $\FOconn_1$, but presumably not in $\TC^1_0$.
\begin{conjecture}
$2$-connectivity cannot be expressed in $\TC^1_0$.
\end{conjecture}

\begin{figure}[ht]
\centering
\begin{tikzpicture}[scale=.9]
\node (fo) at (-2.7,0) {$\FO$};
\node (fo0) at (0,0) {$\FOconn_0$};
\node (fo1) at (3,0) {$\FOconn_1$};
\node (fodots) at (6,0) {$\dots$};
\node (fok) at (9,0) {$\FOconn_k$};

\node (dp1) at (0,2) {$\FODP_1$};
\node (dp2) at (3,2) {$\FODP_2$};
\node (dpdots) at (6,2) {$\dots$};
\node (dpk) at (9,2) {$\FODP_{k+1}$};

\node (tc0) at (0,-2) {$\TC^1_0$};
\node (tc1) at (3,-2) {$\TC^1_1$};
\node (tcdots) at (6,-2) {$\dots$};
\node (tck) at (9,-2) {$\TC^1_k$};

\node (mso) at (11.5,0) {$\MSO$};

\node (v0) at (-1.7,0) {$\subsetneq$};
\node (v1) at (1.5,0) {$\subsetneq$};
\node (v2) at (4.5,0) {$\subsetneq$};
\node (v3) at (7.5,0) {$\subsetneq$};
\node (v31) at (10.5,0) {$\subsetneq$};
\node[rotate=35] (v31) at (10.5,-1) {$\subsetneq$}; % alex : other connection with MSO
\node[rotate=-35] (v31) at (10.5,1) {$\subsetneq$};

\node (v4) at (1.5,-2) {$\subseteq$};
\node (v5) at (4.5,-2) {$\subseteq$};
\node (v6) at (7.5,-2) {$\subseteq$};
\node[rotate=-90] (v7) at (0,-1) {$\subsetneq$};
\node[rotate=-90] (v8) at (3,-1) {$\subsetneq$};
% \node[rotate=35] (v8) at (1.5,-1) {$\overset{?}{\subset}$};
\node[rotate=215] (v8) at (1.5,-1) {$\subset$};
\node[rotate=35] (v81) at (1.3,-0.8) {\scriptsize ?};
\node[rotate=-90] (v9) at (9,-1) {$\subsetneq$};

\node (v10) at (1.5,2) {$\subsetneq$};
\node (v11) at (4.5,2) {$\subsetneq$};
\node (v12) at (7.5,2) {$\subsetneq$};
\node[rotate=90] (v13) at (0,1.2) {$\equiv$};
\node[rotate=90] (v14) at (3,1.2) {$\subsetneq$};
\node[rotate=90] (v15) at (9,1.2) {$\subsetneq$};
\end{tikzpicture}
\caption{Connections between the logics}
\label{fig:connection}
\end{figure}
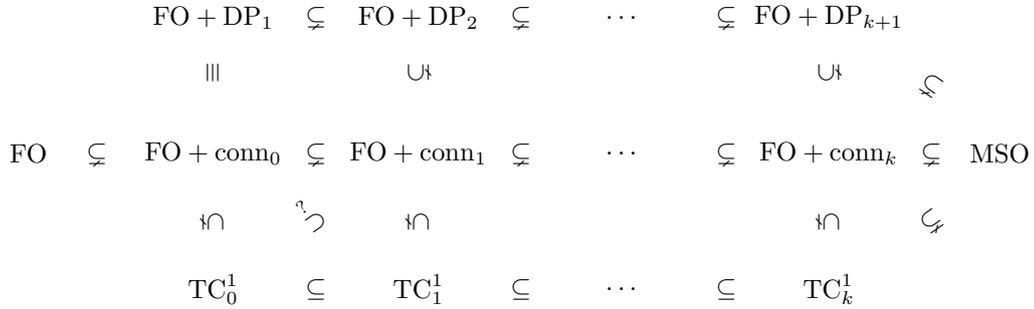

% !TEX root = main.tex

\section{Conclusion}
\label{sec:conclusion}

We studied first-order logic enriched with connectivity
predicates tailored to express algorithmic graph properties that
are commonly studied in contemporary parameterized algorithmics.
This yielded separator logic, which can query connectivity after the deletion
of a bounded number of elements, and disjoint-paths logic, which can
express the disjoint-paths problem. We demonstrated a rich expressiveness
that arises from the interplay of these predicates with the nested
quantification of first-order logic. We also studied the limits of
expressiveness of these new logics.

In a companion paper, we studied the model-checking problem
for separator logic and proved that it is fixed-parameter tractable
parameterized by formula size on classes of graphs that exclude a
fixed topological minor~\cite{pilipczuk2021algorithms}. This yields a powerful algorithmic meta-theorem
for separator logic. On the other hand, while the
disjoint-paths problem is fixed-parameter
tractable on general graphs~\cite{robertson1995graph}, it is not clear that the model-checking
problem for disjoint-paths logic is fixed-parameter tractable beyond
graphs of bounded treewidth. This remains a challenging question
for future work.

It will also be interesting to study other extensions
of first-order logic that can express further interesting algorithmic
graph problems, such as reachability with regular
paths queries. This would, in the simplest case, allow to express
bipartiteness and the odd cycle transversal problem.
On the other hand, it is very likely that with general
regular paths queries, we will get intractability beyond
bounded treewidth graphs.

%%
%% Bibliography
%%

%% Please use bibtex,

\bibliography{ref}

\begin{thebibliography}{10}

\bibitem{agrawal2021fpt}
Akanksha Agrawal, Lawqueen Kanesh, Fahad Panolan, M.~S. Ramanujan, and Saket
  Saurabh.
\newblock An {FPT} algorithm for elimination distance to bounded degree graphs.
\newblock In {\em 38th International Symposium on Theoretical Aspects of
  Computer Science (STACS 2021)}. Schloss Dagstuhl-Leibniz-Zentrum f{\"u}r
  Informatik, 2021.

\bibitem{bojanczyk2021separator}
Miko{\l}aj Boja{\'n}czyk.
\newblock Separator logic and star-free expressions for graphs.
\newblock {\em arXiv preprint arXiv:2107.13953}, 2021.

\bibitem{bonnet2021twin}
{\'E}douard Bonnet, Ugo Giocanti, Patrice Ossona~de Mendez, Pierre Simon,
  St{\'e}phan Thomass{\'e}, and Szymon Toru\'nczyk.
\newblock Twin-width {IV}: ordered graphs and matrices.
\newblock {\em arXiv preprint arXiv:2102.03117}, 2021.

\bibitem{bulian2017parameterized}
Jannis Bulian.
\newblock Parameterized complexity of distances to sparse graph classes.
\newblock Technical report, University of Cambridge, Computer Laboratory, 2017.

\bibitem{bulian2016graph}
Jannis Bulian and Anuj Dawar.
\newblock Graph isomorphism parameterized by elimination distance to bounded
  degree.
\newblock {\em Algorithmica}, 75(2):363--382, 2016.

\bibitem{buro2000simple}
Michael Buro.
\newblock Simple amazons endgames and their connection to {Hamilton} circuits
  in cubic subgrid graphs.
\newblock In {\em International Conference on Computers and Games}, pages
  250--261. Springer, 2000.

\bibitem{courcelle1990monadic}
Bruno Courcelle.
\newblock The monadic second-order logic of graphs. {I}. recognizable sets of
  finite graphs.
\newblock {\em Information and computation}, 85(1):12--75, 1990.

\bibitem{courcelle2000linear}
Bruno Courcelle, Johann~A. Makowsky, and Udi Rotics.
\newblock Linear time solvable optimization problems on graphs of bounded
  clique-width.
\newblock {\em Theory of Computing Systems}, 33(2):125--150, 2000.

\bibitem{book15parameterized-algorithms}
Marek Cygan, Fedor~V. Fomin, Lukasz Kowalik, Daniel Lokshtanov, D{\'{a}}niel
  Marx, Marcin Pilipczuk, Micha{\l} Pilipczuk, and Saket Saurabh.
\newblock {\em Parameterized Algorithms}.
\newblock Springer, 2015.

\bibitem{ebbinghaus2005finite}
Heinz-Dieter Ebbinghaus and J{\"o}rg Flum.
\newblock {\em Finite model theory}.
\newblock Springer Science \& Business Media, 2005.

\bibitem{fagin1974generalized}
Ronald Fagin.
\newblock Generalized first-order spectra and polynomial-time recognizable
  sets.
\newblock {\em Complexity of computation}, 7:43--73, 1974.

\bibitem{FominGST20}
Fedor~V. Fomin, Petr~A. Golovach, Giannos Stamoulis, and Dimitrios~M. Thilikos.
\newblock An algorithmic meta-theorem for graph modification to planarity and
  {FOL}.
\newblock In {\em 28th Annual European Symposium on Algorithms, {ESA} 2020},
  pages 51:1--51:17, 2020.

\bibitem{fomin2021parameterized}
Fedor~V. Fomin, Petr~A. Golovach, and Dimitrios~M. Thilikos.
\newblock Parameterized complexity of elimination distance to first-order logic
  properties.
\newblock {\em arXiv preprint arXiv:2104.02998}, 2021.

\bibitem{fomin2020hitting}
Fedor~V. Fomin, Daniel Lokshtanov, Fahad Panolan, Saket Saurabh, and Meirav
  Zehavi.
\newblock Hitting topological minors is {FPT}.
\newblock In {\em Proceedings of the 52nd Annual ACM SIGACT Symposium on Theory
  of Computing}, pages 1317--1326, 2020.

\bibitem{ganian2014lower}
Robert Ganian, Petr Hlin{\v{e}}n{\`y}, Alexander Langer, Jan
  Obdr{\v{z}}{\'a}lek, Peter Rossmanith, and Somnath Sikdar.
\newblock Lower bounds on the complexity of {MSO1} model-checking.
\newblock {\em Journal of Computer and System Sciences}, 80(1):180--194, 2014.

\bibitem{gradel2007finite}
Erich Gr{\"a}del, Phokion~G. Kolaitis, Leonid Libkin, Maarten Marx, Joel
  Spencer, Moshe~Y. Vardi, Yde Venema, and Scott Weinstein.
\newblock {\em Finite Model Theory and its applications}.
\newblock Springer Science \& Business Media, 2007.

\bibitem{grohe2008logic}
Martin Grohe.
\newblock Logic, graphs, and algorithms.
\newblock {\em Logic and automata}, 2:357--422, 2008.

\bibitem{grohe2017deciding}
Martin Grohe, Stephan Kreutzer, and Sebastian Siebertz.
\newblock Deciding first-order properties of nowhere dense graphs.
\newblock {\em Journal of the ACM (JACM)}, 64(3):17, 2017.

\bibitem{hols2019vertexcover}
Eva{-}Maria~C. Hols, Stefan Kratsch, and Astrid Pieterse.
\newblock Elimination distances, blocking sets, and kernels for vertex cover.
\newblock In {\em {STACS}}, 2020.

\bibitem{jansen2021vertex}
Bart M.~P. Jansen, Jari J.~H. de~Kroon, and Micha{\l} W{\l}odarczyk.
\newblock Vertex deletion parameterized by elimination distance and even less.
\newblock In {\em Proceedings of the 53rd Annual ACM SIGACT Symposium on Theory
  of Computing}, pages 1757--1769, 2021.

\bibitem{kreutzer2010lower}
Stephan Kreutzer and Siamak Tazari.
\newblock Lower bounds for the complexity of monadic second-order logic.
\newblock In {\em 2010 25th Annual IEEE Symposium on Logic in Computer
  Science}, pages 189--198. IEEE, 2010.

\bibitem{libkin2013elements}
Leonid Libkin.
\newblock {\em Elements of finite model theory}.
\newblock Springer Science \& Business Media, 2013.

\bibitem{LindermayrSV20}
Alexander Lindermayr, Sebastian Siebertz, and Alexandre Vigny.
\newblock Elimination distance to bounded degree on planar graphs.
\newblock In {\em 45th International Symposium on Mathematical Foundations of
  Computer Science, {MFCS} 2020, August 24-28, 2020, Prague, Czech Republic},
  pages 65:1--65:12, 2020.

\bibitem{pilipczuk2021algorithms}
Micha{\l} Pilipczuk, Nicole Schirrmacher, Sebastian Siebertz, Szymon
  Toru\'nczyk, and Alexandre Vigny.
\newblock Algorithms and data structures for first-order logic with
  connectivity under vertex failures.
\newblock {\em arXiv preprint arXiv:2111.03725}, 2021.

\bibitem{reed2004finding}
Bruce Reed, Kaleigh Smith, and Adrian Vetta.
\newblock Finding odd cycle transversals.
\newblock {\em Operations Research Letters}, 32(4):299--301, 2004.

\bibitem{robertson1995graph}
Neil Robertson and P.~D. Seymour.
\newblock Graph minors. {XIII}. the disjoint paths problem.
\newblock {\em J. Combin. Theory Ser. B}, 63:65--110, 1995.

\bibitem{DBLP:reference/hfl/Thomas97}
Wolfgang Thomas.
\newblock Languages, automata, and logic.
\newblock In Grzegorz Rozenberg and Arto Salomaa, editors, {\em Handbook of
  Formal Languages, Volume 3: Beyond Words}, pages 389--455. Springer, 1997.
\newblock \href {https://doi.org/10.1007/978-3-642-59126-6\_7}
  {\path{doi:10.1007/978-3-642-59126-6\_7}}.

\end{thebibliography}

\appendix

\end{document}